\documentclass[onefignum,onetabnum]{siamart190516}
\usepackage{amsmath}
\usepackage{hyperref}
\usepackage{lipsum}
\usepackage{amsfonts}
\usepackage{graphicx}
\usepackage{epstopdf}
\usepackage{algorithmic}
\ifpdf
  \DeclareGraphicsExtensions{.eps,.pdf,.png,.jpg}
\else
  \DeclareGraphicsExtensions{.eps}
\fi

\newsiamremark{remark}{Remark}
\newsiamremark{hypothesis}{Hypothesis}
\crefname{hypothesis}{Hypothesis}{Hypotheses}
\newsiamthm{claim}{Claim}

\headers{A Novel and Efficient Algorithm to Solve Subset Sum Problem}{B.Sinchev, A.B.Sinchev, J.Akzhanova, A.M.Mukhanova and Y.Issekeshev}

\title{A Novel and Efficient Algorithm to Solve Subset Sum Problem\thanks{Submitted to the editors 03/14/20.
}}

\author{B.Sinchev\thanks{International University of Information Technology, Almaty, Kazakhstan 
  (\email{b.sinchev@iitu.kz}).}
\and A.B.Sinchev\thanks{National Information Technologies JSC, Astana, Kazakhstan 
  (\email{askar.sinchev@nitec.kz}).}
\and J.Akzhanova\thanks{TOO Astana LRT, Astana, Kazakhstan 
  (\email{zyekudayeva@gmail.com}).}
 \and A.M.Mukhanova\thanks{Almaty University of Technology, Almaty, Kazakhstan 
  (\email{a.mukhanova@atu.kz}).}
  \and Y.Issekeshev\thanks{ISS Corporation 
  (\email{yissekeshev@gmail.com}).}}

\usepackage{amsopn}

\ifpdf
\hypersetup{
  pdftitle={A Novel and Efficient Algorithm to Solve Subset Sum Problem},
  pdfauthor={B.Sinchev, A.B.Sinchev, J.Akzhanova, A.M.Mukhanova, Y.Issekeshev}
}
\fi

\begin{document}

\maketitle

% REQUIRED
\begin{abstract}
    In this paper we suggest analytical methods and associated algorithms for determining the sum of the subsets $X_m$ of the set $X_n$ (subset sum problem). Our algorithm has time complexity $T=O(C_{n}^{k})$ ($k=[m/2]$, which significantly improves upon all known algorithms. This algorithm is applicable to all NP-complete problems. Moreover, the algorithm has memory complexity $M=O(C_n^k)$, which makes our algorithm applicable to real-world problems. At first, we show how to use the algorithm for small dimensions $m=4 ,5 ,6 ,7 ,8$. After that we establish a general methodology for $m>8$. The main idea is to split the original set $X_n$ (the algorithm becomes even faster with sorted sets) into smaller subsets and use parallel computing. This approach might be a significant breakthrough towards finding an efficient solution to $NP$-complete problems. As a result, it opens a way to prove the $P$ versus NP problem (one of the seven Millennium Prize Problems). 
\end{abstract}

% REQUIRED

% REQUIRED

\section{Introduction}
One of the most important problems in computer science is the equality of classes $P$ and NP. This problem was formulated in 1971 and still remains unresolved. A prize of one million US dollars has been awarded for proving the statement $P = NP$ or for proving its refutation. If, nevertheless, it turns out that $P = NP$, then this will make it possible to quickly and efficiently solve many currently unsolvable problems. So what is the nature of the problem? Class $P$ (Polynomial time) includes simple tasks that can be solved in polynomial time. An example of a simple search task is the sorting of an array of arbitrary data, which is necessary to quickly find certain data. Depending on the selected sorting algorithm, for an array of length $n$, the number of operations performed will be from $n \log n$ to $n^2$. Thus, the simplicity of algorithms belonging to the class $P$ lies in the fact that as $n$ increases, the execution time increases slightly. The class $NP$ (Non-deterministic Polynomial time) includes tasks for which the verification of already found solutions can be carried out in polynomial time. The complexity of the problems of this class lies in the fact that directly finding a solution requires significantly more than polynomial time. 

One of the first fundamental reviews of information retrieval problems, which are reduced to the problem of the sum of subsets (subset sum problem), and search engines was presented in \cite{1}. There are important practical and theoretical problems: Knapsack problem, Graph coloring problem, Subset sum problem, Travelling salesman problem, Boolean satisfiability problem and other. All indicated problems belong to the class of $NP$-complete problems. $NP$-complete problems are solved on the basis of exponential and polynomial algorithms.

The problem of finding an $m$-dimensional subset of an $n$-dimensional set ($m\leq n$) whose sum of elements is equal to some given number $S$ is $NP$-complete. The computational complexity of this task depends on two parameters - the number of elements in the original set ($n$) and the accuracy ($p$), defined as the number of binary bits in the numbers that make up the set. It is obvious that the number of subsets of a set containing n elements is $2^n$. 

In \cite{3}, \cite{4}, the time complexity of the algorithm and the memory complexity are determined in the form of an exponential function of the parameter $n$. Schreppel and Shamir in \cite{4} made a significant contribution to the development of search methods based on tabular m-sums. As a practical problem, the knapsack problem is considered. It is shown that the problem is most difficult when $n$ is of high order. The task becomes easy only with very small values of the parameters $n$ and $p$. If $n$ (the amount of input data) is small, then an exhaustive search is quite acceptable. If the parameter $p$ (the number of bits in the numbers of the set) is small, you can use dynamic programming to solve the knapsack problem. It is easy to make sure that the solution is checked quickly (you just need to sum the elements of the found subset). But exponential time must be spent on finding a solution, since it is necessary to check all possible subsets. This makes the solution of the problem posed impractical for large values of $n$. 

\section{Open problems of NP-complete problems}
Many works are devoted to the development of exponential algorithms for solving the subset sum problem, in which the algorithm operation time $T=O(2^{n/2})$ and the required memory $M=O(2^{n/4})$ depend on $n$ and do not allow practical results for large $n$. The main disadvantage of the known tabular methods (sums) is the construction of each row of the table according to the property defined by each keyword. This means that we are obliged to carry out preliminary work on some structuring of the input data. In turn, when using the vector data model, an additional problem arises of dividing the vector space into subspaces according to each keyword. Recently, in \cite{14} and \cite{15} authors tried to come up with new and interesting approaches of solving subset sum problem, but they were still an expensive in terms of time and memory complexities.

Let us consider one of the statements of $NP$-complete problems. 

Theoretical problem (subset sum problem). Given a set of $n$ natural numbers and a number $S$, it is required to find out if there is one or more subsets, each of which consists of m elements ($m\leq n$) and the sum of these elements is $S$. 

It is important to note that the solution of the problem will provide the solution of numerous problems stemming from $NP$-complete problems. 

Now we can move on to the mathematical formulation of the problem of the sum of the subsets and its solution. 

\textbf{The main task.} Given a set of  integers $(x_1, x_2,..., x_n)
 \in X^n$ of dimension $n$, it is required to find out if there exists a subset 
$X_m$ of dimension m such that the following conditions are satisfied: 
\begin{align}\label{eq:subset}
X_m=  &\{x_i+x_j+...+x_g+x_h=S, i \neq j \neq ...\neq  g \neq h, x_i,x_j,...,x_g,x_h\in X^n  ,\\& (i,j,...g,h )\in N=(1,2,...,n),   m\leq n\} \notag
\end{align}

Here $x_i,x_j,...,x_g,x_h\in X_m$, with the number of elements  $x_i,x_j,...,x_g,x_h $ equal to $m$.
 
We introduce the following notation: $C ^m_n$ - combination and $S ^m_n$ - sum of elements
of one subset from the set of subsets $X_m$ of the set $X_n$. The variable $m$ can vary from $0, 1, 2, ...n$. The set of these subsets $X_m$ is determined based on the combination 
\begin{equation}\label{eq:combination}
    C^m_n=\frac{n!}{m!(n-m)!}. 
\end{equation}

Sort the set $X^n$ in increasing order and find the quantities 
\begin{equation}\label{eq:S_min}
    S^m_{min}=\Sigma^m_1 x_i, 
\end{equation}

\begin{equation}\label{eq:S_max}
    S^m_{max}=\Sigma^n_{n-m+1} x_i.
\end{equation}

We compose possible ranges of membership of $S$, the corresponding subset of the set of subsets $X_m$, 
\begin{align}\label{eq:S}
    S \in [S^m_{min}, S^m_{max}].
\end{align}

Moreover, all possible values of the range \ref{eq:S} are calculated as the sum of the elements of the original set, the indices of which are generated based on the combination generation algorithm \ref{eq:combination}.

\section{Methods for solving the problem of the sum of subsets}

The solution to the problem of the sum of subsets is based on \cite{6}, \cite{7} and following theorems. It is assumed that $X_n$ consists of odd and even integers (or natural) numbers. 
\begin{theorem}\label{thm:1}
Let the number $S$ belongs to the range $[S^m_{min}, S^m_{max}]$ and there is a number in this range which equals to $S$. Then there exists a subset $X_m$ whose sum of elements is $S$. 
\end{theorem}

\begin{proof}
Fulfillment of the hypothesis of the theorem (or condition \ref{eq:T}) means that it is necessary to generate all subsets of $X_m$ based on the formula \ref{eq:combination} of $X^n$, where the sum of the elements of each subset varies from the minimum value $S^m_{min}$ to the maximum value $S^m_{max}$. This is equivalent to generating all $n$-dimensional binary vectors from zeros and ones 
($e \in E^n $). The above condition allows you to enumerate the subsets in the order of the minimum change of the binary code of the binary vector $e$. If the $i$-th index of the vector $e$ is $1$ (one), this means that this element is included in this subset and should be taken into account when calculating the sum of the elements. The definition of $m$ indices on which units stand uniquely determines the vector $e$ corresponding to one subset of the set of subsets $X_m$. If there is a number $S$ from the specified range, then there is a binary vector $e$ such that the sum $S$ is calculated on the basis of the scalar product: 
$$S=(e,x).$$
\end{proof}

 We extend the result to a set of subsets $X_{n-m}$ from the set $X^n$ for the range $[S^m_{min}, S^m_{max}]$.

\begin{theorem}\label{thm:2}
Let the number $S$ belongs to the range $[S^m_{min}, S^m_{max}]$ and there is a number in this range which equals to $S$. Then there exists a subset $X_{n-m}$ whose sum of elements is $S$. 
\end{theorem}

\begin{proof}
Based on the equality of combinations $C^m_n=C^{n-m}_n$ we can replace the variable $m$ with the variable $n-m$ and the binary vector $e$ with the binary vector $\bar{e}$ from theorem \ref{thm:1}, in which the zeros of the vector $e$ are replaced by ones and the ones by zeros. Then there exists a binary vector $\bar{e}$ such that $S$ is calculated on the basis of the scalar product: $S=(\bar{e},x)$ or $S=S^n_{min} - (e,x)$. 
\end{proof}

\textbf{Remark 1.} \\
1. If there is an element $x \in X^n$ such that the difference $S-x$ is representable as the sum of $m-1$ elements from the set $X^n$, then there exists a subset $X_{m-1}$.\\
2. It is assumed that in theorem \ref{thm:1}, the number $S$ can vary from $S^m_{min}$ to $S^n_{max}$.\\
3. The dimension of the subset $X_m$ easily extends to $n$ if other elements of this subset are considered zeros, except for elements with indices $(i,j,...,g,h ) \in N.$\\
4. The indicated indices are generated based on the combination generation algorithm 
$C^m_n$ \cite{5}.

The theorems and symmetry of the combination function \ref{eq:combination} determine the maximum value of the parameter $m = n / 2$. 

It is easy to find the running time of an algorithm based on theorem \ref{thm:1}, 
\begin{equation}\label{eq:T}
    T=O(C^m_n).
\end{equation}

Memory $M=O(n)$ is required to store the feasible binary vector $e$. The indices of nonzero components of this vector are determined based on the combination generation algorithm. 

A new approach is proposed for solving the problem of the sum of subsets, based on \ref{eq:T}, \ref{eq:k}. 

\textbf{Case 1.} The variable $m$ can take even values of $4$, $6$, $8$, $10$ or more. In \cite{6}, the cases $m = 2$ and $m = 3$ were considered in detail. Then the variable $k$ is determined by the formula 
\begin{equation}\label{eq:k}
    k=m/2.
\end{equation}

We construct a subset $Z^l=\{ z_1,z_2,...,z_l \} $, consisting of the sum of elements $x_i$ with indices determined on the basis of the algorithm for generating the combination $C^k_n$, from the set $X^n$, $l=C^k_n$.
In particular, if $X^7=\{ x_1,x_2,x_3,x_4,x_5,x_6,x_7 \}$, $n=7$, $k=3$, $l=35$, then the subset $Z^{35}=\{z_1,z_2,...,z_{35}\}$ consists of elements $z_1=x_1+x_2+x_3$, $z_2=x_1+x_2+x_4$, $z_3=x_1+x_2+x_5$, ... , $z_{34}=x_4+x_6+x_7$, $z_{35}=x_5+x_6+x_7$. Each element is the sum of three elements with indices determined based on the algorithm for generating the combination $C^3_7$ from the set $X^7$.

We introduce the mapping of the subset $Z^l$ into the set $Y^l$: 
\begin{equation}\label{eq:map}
    y=\tau(S,z)=(S - z)z,  \forall z \in Z^l.
\end{equation}

Based on it, we obtain 

\begin{equation}\label{eq:Yl}
    Y^l=\{ y_1 y_2... y_l \Longleftrightarrow \tau(S,z_i)=y_i,  z_i \in Z^l, i=1,2,...,l \}.
\end{equation}

Among the set $Y^l$, let there exist elements such that the identity holds: 

\begin{equation}\label{eq:Yl2}
    y_i=y_j, i\neq j, j \in L, L=\{ 1,2,...,l \}
\end{equation}

\begin{lemma} \label{lemma:1}
Let the number $S$ belong to the range $[S^m_{min}, S^m_{max}]$ and identity \ref{eq:Yl2} hold for the set \ref{eq:Yl}. Then there are one or more subsets of $X_m$ and the main problem is solvable. 
\end{lemma}

\begin{proof}
The first condition shows the existence of the subset $X_m$ from theorem \ref{thm:1}. To construct a subset $X_m$ satisfying a given $S$, based on identity \ref{eq:Yl2}, we have $y_i=\tau(S,z_i)=(S-z_i)z_i=z_jz_i$, if $z_j=S-z_i$. Here, the quantities $z_i$, $z_j$ are the sum of $k$ elements from the set $X^n$, whose indices are determined based on the combination generation algorithm $C^k_n$. On the other hand, according to the map \ref{eq:map}, $y_j=\tau(S,z_j)=(S-z_j)z_j=z_i z_j$, similarly assuming that $z_i=S-z_j$. In fact, the quantities $z_i,z_j$ are the roots of the quadratic equation $z^2-Sz+c=0$
. According to the Vieta's formula, $c=z_i z_j$. Thus, we obtain $y_i=y_j=z_iz_j$. That's why identity \ref{eq:Yl2} holds. Then there exist elements
$z_i$ and  $z_j$ such that $z_i+z_j=S$,   $X_m=X^i_k \bigcup X^j_k$.  The subsets 
$X^i_k$, $X^j_k$ are formed on the basis of the quantities $z_i$, $z_j$ in which addition operations are excluded. The inequality $i \neq j$ from \ref{eq:Yl2} ensures the union of these subsets with diverging indices. 
\end{proof}

Based on the lemma \ref{lemma:1}, the time complexity is $T=O(l)$ and the memory complexity is $M=O(l)$.   

\textbf{Remark 2.} For $k = 2, 3, 4$ there are subsets $X_4, X_6, X_8$ with parameters $m = 4$, $m = 6$, $m = 8$, respectively. Thus, each element $z_i \in Z^l$ is a sum of $k$ elements $x$ with non-coincident indices from the set $X^n$ determined by the algorithm for generating the combination $C^k_n$. Under condition \ref{eq:Yl2}, there exist elements $y_i,y_j$ from the subset $Y^l$ with indices $i, j$ and such that $i \neq j$ such that condition $X^i_k \bigcup X^j_k \neq \emptyset $ is fulfilled. 

\textbf{Example 1}. It is required to find out if there exists a subset $X_6$
 consisting of six different elements of the set $X^7$, the sum of which is $S$. According to the lemma \ref{lemma:1}, the elements $y_i=y_j$ must exist. Since $k = 3$, then each of the variables $y_i,y_j$ consists of three elements from the set $X^7$ with diverging indices. These indices are determined based on the combination generation algorithm $C^3_7$. In particular, for this subset $Y^{35}$, the following coincidences are possible: 
$y_1=y_{34}$  or $y_2=y_{35}$ or $y_1=y_{35}$ or $y_2=y_{35}$ and other combinations. Let the first and penultimate conditions be satisfied, and also the second and last conditions. Then we have the subsets $X_6=\{x_1,x_2,x_3,x_4,x_6,x_7\}$, $X_6=\{x_1,x_2,x_4,x_5,x_6,x_7\}$. 

\textbf{Case 2.} The variable m takes odd values: $5,7,9,11$ and further. Then the variable $k$ is determined by the formula $k = (m-1) / 2$. 

We introduce the quantity

\begin{equation}\label{eq:Sl}
    S(x_{\bar{l}})=S-x_{\bar{l}}, \forall x_{\bar{l}} \in X^n, \bar{l} \in \mathbb{N}.
\end{equation}

\begin{lemma}\label{lemma:2}

Let the number $S$ belong to the range $[S^m_{min},  S^m_{max}]$, and for some element $x_{\bar{l}} \in X^n$, and taking into account formula \ref{eq:Sl}, identity \ref{eq:Yl2} holds. Then there are one or more subsets of $X_m$ and the main problem is solvable. 
\end{lemma}

\begin{proof}
The first condition shows the existence of the subset $X_m$ from theorem \ref{thm:1}. To construct a subset $X_m$ satisfying the number $S$, based on identity \ref{eq:Yl2}, and taking into account formula \ref{eq:Sl}, we have $\tau(S(x_{\bar{l}}),z_i)=(S(x_{\bar{l}})-z_i)z_i=z_jz_i$, if $S(x_{\bar{l}})-z_i=z_j$. On the other hand, $\tau(S(x_{\bar{l}}),z_i)=(S(x_{\bar{l}})-z_j)z_j=z_iz_j$,  similarly assuming that $S(x_{\bar{l}})-z_j=z_i$. It means that the conditions of lemma \ref{lemma:1} are satisfied. Then we have $zi+zj+x_{\bar{l}}=S$. Here $k=(m-1)/2$, $X_m=X^i_k \bigcup X^j_k \bigcup x_{\bar{l}}$. The subsets $X^i_k$,$X^j_k$ are formed on the basis of the quantities $z_i$, $z_j$ in which addition operations are excluded. 
\end{proof}

\textbf{Remark 3.} The combination is determined by the formula $C^k_{n-1}$,  $l=C^k_{n-1}$. The element $x_{\bar{l}}$ is excluded from the set $X^n$ and the set $X^{n-1}$ is considered. 

The running time of an algorithm based on lemma \ref{lemma:2} is determined by the formula $T=O(nl)$ according to remark 3. The required memory $M=O(l)$ is needed to preserve the subset $Y^l$.
 
\textbf{Example 2.} It is required to find out if there exists a subset $X_5$ consisting of five different elements of the set $X^7$, the sum of which is $S$. According to lemma \ref{lemma:2}, there must exist elements $y_i=y_j$. Since $k = 2$ and $x_{\bar{l}}=x_7$, $l=C^2_6$,  $Y=\{ y_1,y_2,...,y_{14},y_{15} \}$. Then each of the variables $y_i,y_j$ consists of two elements from the set $X^6$ with diverging indices. These indices are determined based on the combination generation algorithm $C^2_7: (1,2), ..., (1,6), (2,3), ..., (2,6), ..., (5,6)$ . In particular, for the subset $Y^{15}$, the following coincidences are possible: $y_1=y_{14}$  or $y_2=y_{15}$ or $y_1=y_{15}$ or $y_2=y_{15}$ and other combinations. Let the first and penultimate conditions be satisfied, and also the second and last conditions. Then we have the subsets $X_5=\{x_1,x_2,x_4,x_6,x_7\}$, $X_6=\{x_1,x_3,x_5,x_6,x_7\}$. If condition \ref{eq:Yl2} is not fulfilled, another element $x_{\bar{l}}$ is selected and this process is repeated until condition \ref{eq:Yl2} is satisfied. 

The running time \ref{eq:T} of possible algorithms built on theorem \ref{thm:1} and theorem \ref{thm:2}, as the variable $m$ approaches the value $n/2$, increases rapidly. Therefore, for large values of the parameter k, the dimension $l=C^k_n$ of the subset $Y^l$ increases and the selection of indices $i$, $j$ becomes more complicated. 

\textbf{Example 3.} Suppose we have a set $X^{16}=\{17, 2, 3, 23, 19, 1, 14, 20, 6, 10, 4, 25, 7$, \\
$49, 41, 5\}$. It is required to find a subset $X_6$ with a given sum $S=137$.

According to the theorem \ref{thm:1}, $S=137 \in [21,177]$ and $m=6$. We rewrite the original set in the form $X^{16}=\{x_1,x_2,x_3,...,x_{14},x_{15},x_{16}\}$. The main parameters are $n=16$, $k=m/2=3$, $l=C^k_n=C^3_{16}=560$. We form the subset $Z^{156}=\{z_1,z_2,...,z_{560}\}$, consisting of the elements $z_1=x_1+x_2+x_3$, $z_2=x_1+x_2+x_4$, $z_3=x_1+x_2+x_5$,...,$z_{558}=x_{13}+x_{14}+x_{15}$, $z_{559}=x_{13}+x_{14}+x_{16}$, $z_{560}=x_{14}+x_{15}+x_{16}$. Each element is the sum of three elements with indices determined based on the combination generation algorithm $C^3_{16}$ from the set $X^{16}$. Next, we find the subset $Y^{560}=\{y_1,y_2,...,y_{560}\}$ based on the application of the map \ref{eq:map} to the subset $Z^{560}$. We check the condition $y_i=y_j$, which is satisfied for the indicies $i=2$, $j=560$, because $y_2=(S-z_2)*z_2=(137-42)*42$ and $y_{560}=(S-z_{560})*z_{560}=(137-95)*95$. We form the subset $X_6=X^i_3 \bigcup X^j_3$ on the basis of the subsets $X^i_3$ and $X^j_3$, which correspond to the elements $z_2$ and $z_{560}$ with excluded additional operations. The indices of these subsets are $i=2$, $j=560$. Thus, we have $X_6=\{ x_1,x_2,x_4,x_{14},x_{15},x_{16}\}$ and $S=137$.  

This proposition is illustrated by remark 2 and numerical example 3. Then we partition the set $X^n$ into $\bar{n}$ disjoint subsets 
$X^{m_1},X^{m_2},...,X^{m_{\bar{n}}}$ with respect to indices whose dimensions satisfy the conditions $m_1+m_2+...+m_{\bar{n}}=n$,     $m_i>>k=[m/2]$, $i=1, 2, 3,...,\bar{n}$. To each subset $X^{m_i}$, lemma \ref{lemma:1} and lemma \ref{lemma:2} are applicable. 

Let the quantities $k_1$, $k_2$ depend on the parameter $m$, with $m=k_1+ k_2$. It is assumed  $k_1= k_2$. It is believed that for each subset $X^{m_i}$ the original problem is solved.  

\begin{lemma}\label{lemma:3}
If there are two subsets $X^i_{k_1} \in X^{m_i}$ and $X^j_{k_2} \in X^{m_j}$, selected from the subsets $X^{m_1},X^{m_2},...,X^{m_{\bar{n}}}$ based on the combination $C^2_{\bar{n}}$ and satisfying the main conditions of lemma \ref{lemma:1} and such that $X_m=X^i_{k_1} \bigcup X^j_{k_2}$ and the indices of the elements of these subsets $X^i_{k_1}$,  $X^j_{k_2}$ are determined based on the generation algorithms for the combinations $C^{k_1}_{m_i}$ and $C^{k_2}_{m_j}$ respectively, then there are one or more subsets of  $X_m$ and the main problem is solvable. 
\end{lemma}

\begin{proof}
Initially, we decompose the original set $X^n$ into the subsets $X^{m_i}$. Next, we make a sampling using lemma \ref{lemma:1} of the subset $X^i_{k_1}$, whose elements belong to the subset $X^{m_i}$. We carry out similar actions for the subset $X^j_{k_2}$ from the subset $X^{m_j}$. Then the subset $X_m=X^i_{k_1} \bigcup X^j_{k_2} \neq \emptyset$, the sum of the elements of which is equal to $S$ and at the same time $X^i_{k_1} \bigcap X^j_{k_2}= \emptyset $
 due to the conditions for partitioning the set $X^n$.
\end{proof}

\textbf{Remark 4.} If $X_m=X^i_{k_1} \bigcup X^j_{k_2} \neq \emptyset$, then other combinations based on the combination $C^2_{\bar{n}}$ are considered. If $k_1 \neq k_2$, the subsets $X^i_{k_2}$, $X^j_{k_1}$ are additionally defined for the possible construction of another subset $X_m=X^i_{k_2} \bigcup X^j_{k_1}$. 

\textbf{Remark 5.}  The conditions of the lemma \ref{lemma:3} are true when combining a larger number of subsets, for example, $X_m=X^i_{k_1} \bigcup X^j_{k_2} \bigcup X^k_{k_3}$, $m=k_1+k_2+k_3$, and the selection of three subsets of the subsets $X^{m_1},X^{m_2},...,X^{m_{\bar{n}}}$ is carried out on the combination $C^3_{\bar{n}}$
 and so on. The number of samples depends on the parameter $m$ of the subset 
$X_m$.
 
The operating time of the algorithm based on lemma \ref{lemma:3} is determined by the formula 
$O(C^k_{m_j}) \leq T \leq O(C^{k_1}_{m_i} C^{k_2}_{m_j})$, if $m_i<m_j$. The required memory is $M=O(C^{k_1}_{m_i}+C^{k_2}_{m_j})$.
  
\section{Discussion of the results} 
In the scientific literature, there are algorithms for solving the problem of the sum of subsets based on exponential algorithms from \cite{4}\cite{5}. The search time and the required memory are $O(2^{n/2})$ and $O(2^{n/4})$  respectively. The main drawback of these algorithms is that the required time and memory for information processing is expressed through the exponent, namely, the time multiplied by the memory $T*M=O(2^n)$ ($n$ is the amount of processed data). The last remark imposes very strict requirements on the hardware and other means of information processing. Let us determine the applicability limit of the algorithms from \cite{4} \cite{5}. Set $n = 128$. To sort a subset of $2^{64}$  elements, it takes $O(2^{70})$ time. It is known that modern computers and information technologies can work with $2^{66}$ data elements. Thus, even subset sorting is difficult. 

Initially, new basic algorithms were developed for problems on the sum of the subset $X_m$ with dimensions 4, 5, 6 and more from the set $X^n$ ($m = 4$, $m = 5$, $m = 6$ or more) belonging to the set $X^n$, with the above advantage in time $T=O(C^k_n)$ and memory $M=O(C^k_n)$ ($k=k_1= k_2=[m/2]$). This advantage is determined by the ratio $C^m_n/C^k_n$. For comparison, the time exhaustive search of the subset $X_m$ is equal to $T=O(n^m)$. The found time is much shorter than the exhaustive search time by so many times $C^m_n/C^k_n$ (or so many times $n^{m-k}$). 

The partitioning of the original set $X^n$ into subsets proposed as $m \rightarrow n/2$ allows us to solve the problem of the sum of the subset $X_m$ with dimension eight or more belonging to the set $X^n$. Thus, new effective methods and algorithms for solving the problem of the sum of subsets in a network computing environment with optimization in time and memory were found. In particular, with $m = 1$ and $m = n-1$, the following traditional search methods follow from the theorems: sequential search and pattern matching (mask search). 

We emphasize that in exponential algorithms for solving the problem of the sum of subsets, sorting methods are mandatory, which ensure the operating time of the algorithm as $T=O(2^{n/2})$. Sorting takes longer. 

The developed theorems and lemmas make it possible to construct a whole family of algorithms for solving the problem of the sum of subsets. 

\section{Conclusion} Algorithms for solving problems from the NP-complete class are used every day in a large number of areas: when restoring partially damaged files, decomposing a number into simple factors, in cryptography, optimizing various routes and sizes of delivered goods, in logistics, and so on. A much more effective solution to such problems could save serious money, as well as time, because with modern algorithms, we cannot solve fast enough, and we have to be content with only approximate solutions. In addition, in modern medicine, there is no shortage of complex computational problems, and moreover, fast algorithms would bring the analysis of various diseases to a whole new level, which would help save many more lives. 

The results show that the proposed exact algorithms for solving the problem of the sum of the subsets significantly reduce operating time, as well as reduce the hardware requirements for the power of computers, servers and other computing devices. The developed mathematical theory for solving the problem of the sum of the subsets will solve many theoretical and practical problems.

\section*{Acknowledgments}
We would like to acknowledge the assistance of Nurlan Abdukadyrov (PhD student, University of Illinois at Chicago) and Sreenivas (Vas) Vedantam (Juris Doctor, Pattern Attorney, Baker McKenzie) in preparing and checking our work.

\section*{Disclaimer}
Copyrights for components of this work owned by others than the author(s) must be honored. 
\bibliographystyle{siamplain}
\bibliography{references}
\end{document}